\documentclass[submission,copyright,creativecommons]{eptcs}

\providecommand{\event}{AFL 2026} 

\usepackage{iftex}

\ifpdf
  \usepackage{underscore}         
  \usepackage[T1]{fontenc}        
\else
  \usepackage{breakurl}           
\fi

\usepackage{graphicx} 
\usepackage{amsmath, amssymb, amsthm}
\usepackage{tikz}
\usepackage{comment}
\usepackage{lineno}
\usetikzlibrary{automata,positioning,arrows.meta,calc}

\definecolor{stateyellow}{RGB}{219,169,26}
\definecolor{stateteal}{RGB}{34,190,200}
\definecolor{edgepurple}{RGB}{120,0,200}
\definecolor{stateviolet}{RGB}{170,90,210}
\definecolor{edgeturquoise}{RGB}{70,210,200}

\DeclareMathOperator{\ta}{\mathtt{a}}

\usetikzlibrary{automata, positioning, arrows.meta}

\DeclareMathOperator{\N}{\mathbb{N}}

\newtheorem{question}{Question}
\newtheorem{problem}{Problem}
\newtheorem{theorem}{Theorem}
\newtheorem{proposition}[theorem]{Proposition}
\newtheorem{lemma}[theorem]{Lemma}
\newtheorem{corollary}[theorem]{Corollary}

\theoremstyle{definition}
\newtheorem{definition}{Definition}
\newtheorem{example}{Example}
\newcommand{\bigo}{O}

\title{Deterministic Bandwidth of Finite Languages}
\author{Da-Jung Cho\institute{Ajou University}\email{dajungcho@ajou.ac.kr} \and Szil\'ard Zsolt Fazekas\institute{Akita University}\email{szilard.fazekas@ie.akita-u.ac.jp} \and Max Wiedenh\"oft\institute{Kiel University}\email{maw@informatik.uni-kiel.de}}

\def\titlerunning{Deterministic Bandwidth of Finite Languages}
\def\authorrunning{D.-J. Cho et al.}

\begin{document}

\maketitle
\begin{abstract}
Bandwidth restricts how far transitions can move under an ordering of the states in an automaton.
While every finite language admits a bandwidth-$2$ NFA representation, the deterministic setting is substantially more restrictive.
We investigate the bandwidth of partial DFAs accepting finite languages.
We show that bounded bandwidth imposes strong structural restrictions on deterministic representations. We prove that there is an infinite hierarchy of classes of finite languages defined by deterministic bandwidth. As a special case of interest, we consider finite languages accepted by bandwidth-$1$ partial DFAs, and show that they admit a positional characterization, which yields a polynomial-time decision algorithm.
We further study how the minimum DFA bandwidth can be estimated from the structure of the minimal DFA.
We derive computable upper and lower bounds based on position-unfolding, and local growth of reachable residual states.
These bounds can be computed efficiently and differ by at most a linear factor in the maximum word length.
We also present a simple language family where the bounds match exactly.
\end{abstract}

    \section{Introduction}

Representing finite languages by finite automata is a central problem in automata theory~\cite{hopcroft}.
In many applications, however, the goal is not only to recognize a language, but also to obtain a compact and well-structured representation of it. Recently, finite automata encodings of finite languages were studied in the context of template construction problems for sequential generation systems~\cite{ChoFSW25}. The main goal was to construct a single template that generates a desired set of target sequences while avoiding unintended ones. To achieve this, target sequence sets were represented by finite automata and encoded onto circular template structures. In particular, arbitrary regular languages were realized by encoding NFAs into circular DNA templates and simulating their computations through co-transcriptional splicing mechanisms~\cite{ChoFSW25NC,ChoFSW25}. The size and structure of the encoded automaton directly affect the complexity of the resulting template construction. This naturally leads to the study of automata with local transition behavior.
Related structural restrictions on deterministic finite automata have been studied for subregular language classes, including ordered automata~\cite{HolzerT15,ShyrT74, Truthe14}, although they impose different constraints from bandwidth.

One natural structural restriction is bandwidth~\cite{ChinnCDG82,Harper06}. For an ordering of the states, the bandwidth determines the maximum distance that can be traversed by a single transition~\cite{ChoFSW25}. Small bandwidth restricts transitions to nearby states and enforces local transition behavior. Our previous work studied k-bandwidth NFAs~\cite{ChoFISTW26CIAA}. In these automata, transitions can skip only a bounded number of states under a suitable ordering of the automaton. We showed that k-bandwidth NFAs form a strict hierarchy of language classes, where increasing the bandwidth strictly increases expressive power. We also proved that every finite language admits a bandwidth-$2$ NFA representation.

This naturally raises the deterministic counterpart of the problem.
Unlike NFAs, DFAs cannot merge computational paths through nondeterminism~\cite{hopcroft,SakodaS78}. All branching behavior must therefore appear explicitly in the automaton structure. Because of this, bandwidth becomes a more restrictive structural property in deterministic automata. Moreover, minimizing the number of states does not necessarily lead to representations with small bandwidth. A minimal DFA may require substantially larger bandwidth than another equivalent deterministic automaton. This suggests that DFA bandwidth captures a structural aspect of deterministic representations that differs fundamentally from classical state complexity measures.

The deterministic case raises several new questions. Does there exist a constant $k$ such that every finite language can be represented by a k-bandwidth DFA? If not, how does the required bandwidth grow with the structure of the language? Can the minimum bandwidth be estimated from the minimal DFA? More generally, what kinds of finite languages admit low-bandwidth deterministic representations? These questions are fundamentally different from the nondeterministic setting. For NFAs, bandwidth~$2$ already suffices for all finite languages. In DFAs, however, as we will see, we get an infinite hierarchy of language classes defined by their bandwidth, even if we restrict consideration to finite languages.
Similar infinite hierarchies over finite languages have been established for other automata complexity measures, such as union-complexity~\cite{Nagy19}.

In this paper, we study the bandwidth of partial DFAs accepting finite languages. We first investigate the special case of bandwidth~$1$. We show that finite languages accepted by bandwidth-$1$ partial DFAs have a positional structure. This characterization leads to a polynomial-time decision algorithm. We then study upper and lower bounds on DFA bandwidth. For the upper bound, we introduce a construction based on unfolding the minimal partial DFA along word positions. For the lower bounds, we analyze the local growth of reachable residual states in the minimal DFA~\cite{Brzozowski64}. These bounds can be computed efficiently from the minimal DFA of the language. Together, they provide structural estimates for the minimum DFA bandwidth. We further show that the gap between the upper and lower bounds is bounded linearly by the maximum word length. Finally, we present an infinite language family where the bounds match exactly, albeit using an unbounded alphabet.

These results highlight a clear difference between nondeterministic and deterministic automata. For DFAs, bandwidth becomes closely connected to the structural organization of the language and its residual states.

    \section{Preliminaries}

Let $\N$ denote the set of positive integers.
For some $m\in\N$, denote by $[m]$ the set $\{1,2,...,m\}$.
With $\Sigma$, we denote a finite set of symbols, called an \emph{alphabet}.
The elements of $\Sigma$ are called \emph{letters}.
A \emph{word} over $\Sigma$ is a finite sequence of letters from $\Sigma$.
With $\varepsilon$, we denote the \emph{empty word}.
The set of all words over $\Sigma$ is denoted by $\Sigma^*$.
The \emph{length} of a word $w\in\Sigma^*$, i.e., the number of letters in $w$, is denoted by $|w|$; hence, $|\varepsilon| = 0$.
For some $w\in\Sigma^*$, if $w = xyz$ for some $x,y,z\in\Sigma^*$, then $x$ is a \emph{prefix}, $y$ is a \emph{factor}, and $z$ is a \emph{suffix} of $w$. 
A nondeterministic finite automaton~(NFA) is a tuple $A = (\Sigma, Q, q_0,\delta, F)$ where $\Sigma$ is the input alphabet, $Q$ is the finite set of states, $\delta \colon Q \times \Sigma \rightarrow 2^Q$ is the multivalued transition function, $q_0 \in Q$ is the initial state, and $F \subseteq Q$ is the set of final states.
In the usual way, $\delta$ is extended as a function $Q \times \Sigma^* \rightarrow 2^Q$ and the language accepted by $A$ is $L(A) = \{ w \in \Sigma^* \mid \delta(q_0, w) \cap F \neq \emptyset \}$.
The automaton $A$ is a deterministic finite automaton~(DFA) if $\delta$ is a single valued partial function.
It is well known that deterministic and nondeterministic finite automata recognize the class of {\em regular languages}~\cite{hopcroft}.
In a previous paper \cite{ChoFISTW26CIAA}, the authors were concerned with the notion of $k$-bandwidth NFA.

\begin{definition}[$k$-bandwidth NFA]\label{definition:k-bandwidth-NFA}
    Given $k\in\N$, an NFA $A = (Q, \Sigma, \delta, q_0, F)$ is \emph{$k$-bandwidth} ($k$-BW-NFA) if the states can be indexed as $Q = \{q_0,...,q_{n-1}\}$ such that for all $q_i,q_j\in Q$ and $\ta\in\Sigma\cup\{\varepsilon\}$,  $q_j\in \delta(q_i,\ta)$ implies $0 < (j-i) \bmod n \leq k $.
\end{definition}

In this paper, we extend the view to the deterministic setting. For practical reasons, in this extension, we ignore sink states and therefore focus on partial DFAs. This goes hand in hand with the practical motivation of using this model in the context of encoding sets of words on DNA templates in the context of co-transcriptional splicing (see Introduction and references therein). We obtain the following definition.

\begin{definition}[$k$-bandwidth DFA]\label{definition:k-bandwidth-DFA}
    Given $k\in\N$, we say that a partial DFA $A = (Q, \Sigma, \delta, q_0, F)$ is \emph{$k$-bandwidth} ($k$-BW-DFA) if it is also a $k$-bandwidth NFA.
\end{definition}


\subsection{Problems Regarding the Deterministic Bandwidth of Finite Languages}

Similarly to the NFA setting, based on the motivation provided in the introduction, we can ask the following practically motivated questions. First, we know that for all finite languages $L_f$, there exists a $2$-BW-NFA $A$ such that $L(A) = L_f$. As such a small constant bound promises a higher practical feasibility, it would be interesting to know whether such a bound also exists in the deterministic setting.

\begin{question}
    Does there exist a $k\in\mathbb{N}$ such that for all finite languages $L_f$, there is a $k$-BW-DFA $A$ with $L(A) = L_f$?
\end{question}

As we will see later in the paper, the answer to this question is \emph{no}. Hence, determining the bandwidth necessary to model a given finite language $L_f$ in a $k$-BW-DFA emerges as a relevant problem.

\begin{problem}
    Given a finite language $L_f$ and a number $k > 0$, does there exists a $k$-BW-DFA $A$ with $L(A) = L_f$?
\end{problem}

As this problem is likely intractable, given that the bandwidth problem for graphs is NP-hard, even for very restricted graph classes, such as directed acyclic graphs, trees, etc.~\cite{ChinnCDG82}, and even hard to approximate in many cases, here we focus on efficiently computable lower and upper bounds for input finite languages presented as DFA.

\begin{question}
    Let $L_f$ be some finite language and let $k\in\N$ be the minimal bandwidth necessary to construct a $k$-BW-DFA $A$ with $L(A) = L_f$. Can we give efficiently computable nontrivial upper and lower bounds for $k$?
\end{question}

Also, similarly to the NFA setting, we could ask whether there exists a strict hierarchy of classes of languages that can be accepted by $k$-BW-DFA for all $k>0$. As mentioned before, we will see that even for finite languages, no constant $k$ exists that can be used to model all finite languages with $k$-BW-DFA. Hence, we can extend the question of the existence of such a hierarchy even to the finite setting.


\begin{question}
    For every $k\in\N$, does there exist a finite language $L_f$ with a minimum DFA bandwidth of exactly $k$?
\end{question}


\section{The Hierarchy of $k$-Bandwidth DFAs}

In this section, we are concerned with the last question given in the last section, i.e., the question of whether there exists a strict hierarchy of languages recognized by $k$-BW-DFA, based on the selection of $k$. For $k$-BW-NFA, the existence of such a hierarchy has been shown in~\cite{ChoFISTW26CIAA}. Indeed, the construction given there utilizes partial DFAs already. Hence, in general over all regular languages, we immediately obtain the following result. Let $\mathcal{L}_{d,k}$ be the class of languages recognizable by $k$-BW-DFA.

\begin{corollary}\cite{ChoFISTW26CIAA}
    The proper inclusion $\mathcal{L}_{d,k-1} \subsetneq \mathcal{L}_{d,k}$ holds, for any $k \geq 2$, even for binary alphabets $\Sigma$.
\end{corollary}

That leaves us with the question whether such a strict hierarchy also exists among the set of finite languages. Let $\mathcal{L}_{f,k}$ be the class of finite languages recognizable by $k$-BW-DFA. In~\cite{ChoFISTW26CIAA}, we have seen that a bandwidth of $2$ suffices to obtain all finite languages. As the following result shows, this does not hold for $k$-BW-DFA.  Indeed, we obtain the following.

\begin{proposition}[Infinite hierarchy]\label{proposition:dfa-bandwidth-finite-hierarchy}
    The proper inclusion $\mathcal{L}_{f,k-1} \subsetneq \mathcal{L}_{f,k}$ holds, for any $k \geq 2$, even for binary alphabets $\Sigma$.
\end{proposition}
\begin{proof}
    To show this, for each $k \geq 2$, we construct a specific finite language $L_{f,k}$, show that there exists a $k$-BW-DFA $A$ with $L(A) = L_{f,k}$, and then show that there cannot exist any $(k-1)$-BW-DFA $B$ with $L(B) = L_{f,k}$. Let $\Sigma = \{a,b\}$ be a binary alphabet. We set
    \[L_{f,k} = \{\ a^ib^{k^2}a^i \mid i\in[k]\ \}\]
    as the set of all palindromes with $a$-prefixes of lengths $1$ to $k$ and $b$-paths of length $k^2$ in the middle.

    First, we observe that we can construct a $k$-BW-DFA $A = (Q,\Sigma,\delta,q_0,F)$ that accepts $L_{f,k}$. Let $Q = \{q_0,...,q_n\}$ where $n+1 = 1+k+k(k^2+k)$. We obtain that number by the following intuition of the construction: Starting from the initial state (the $1$ in $n$), we use the first $k$ additional states to read the $a^i$ prefixes of each word (the $+k$ component in $n$). Then, we use $k$ alternating distinct paths in the automaton that read the $b^{k^2}a^i$ suffixes, respectively (the $k(k^2+k)$ component in $k$, where $k^2+k$ is the upper bound of the length $|b^{k^2}a^i|$). Each of these paths utilizes the bandwidth $k$ fully and starts in the state reached after reading the respective $a^i$ prefix.\footnote{Clearly, all but one suffix $a^i$ does not have length $k$, hence, this number merely represents an upper bound and could be optimized slightly} Figure~\ref{figure:dfa-hierarchy-construction} illustrates this construction.

    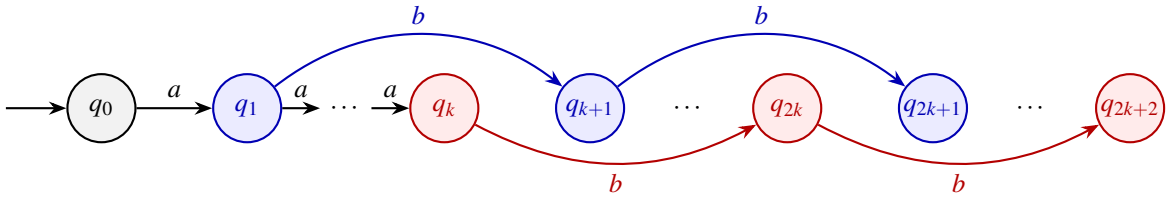
\begin{figure}[htbp]
        \centering
        
        \begin{tikzpicture}[
            >=Stealth,
            every node/.style={font=\small},
            state/.style={
                circle,
                draw,
                minimum size=9mm,
                inner sep=0pt,
                thick
            },
            blue state/.style={
                state,
                draw=blue!70!black,
                text=blue!70!black,
                fill=blue!8
            },
            red state/.style={
                state,
                draw=red!70!black,
                text=red!70!black,
                fill=red!8
            },
            neutral state/.style={
                state,
                draw=black,
                text=black,
                fill=gray!10
            },
            trans/.style={->, thick}
        ]
        
        \node[neutral state] (q0) {$q_0$};
        \node[blue state, right=of q0] (q1) {$q_1$};
        
        \node[right=0.5cm of q1] (dots1) {$\cdots$};
        \node[red state, right=0.5cm of dots1] (qk) {$q_k$};
        
        \node[blue state, right=1.0cm of qk] (qkp1) {$q_{k+1}$};
        
        \node[right=0.5cm of qkp1] (dots2) {$\cdots$};
        \node[red state, right=0.5cm of dots2] (qkl) {$q_{2k}$};
        
        \node[blue state, right=1.0cm of qkl] (q2kp1) {$q_{2k+1}$};
        
        \node[right=0.5cm of q2kp1] (dots3) {$\cdots$};
        \node[red state, right=0.5cm of dots3] (q2l2) {$q_{2k+2}$};
        
        \draw[trans] ($(q0.west)+(-0.8,0)$) -- (q0.west);
        
        \draw[trans] (q0) -- node[above] {$a$} (q1);
        \draw[trans] (q1) -- node[above] {$a$} (dots1);
        \draw[trans] (dots1) -- node[above] {$a$} (qk);
        
        \draw[trans, blue!70!black, bend left=38]
            (q1) to node[above] {$b$} (qkp1);
        
        \draw[trans, red!70!black, bend right=28]
            (qk) to node[below] {$b$} (qkl);
        
        \draw[trans, blue!70!black, bend left=38]
            (qkp1) to node[above] {$b$} (q2kp1);
        
        \draw[trans, red!70!black, bend right=28]
            (qkl) to node[below] {$b$} (q2l2);
        
        \end{tikzpicture}
        
        \caption{Illustration of construction of $k$-BW-DFA for the prefixes $L_{f,k}$. The remaining suffixes are added by extending the parallel $b$-paths until the end of each word and setting the last states to be final, respectively.}
        \label{figure:dfa-hierarchy-construction}
    \end{figure}

Now, suppose towards contradiction that there exists a $(k-1)$-BW-DFA $B = (Q,\Sigma,\delta,q_0,F)$ with $L(B) = L_{f,k}$. Let $h=k-1$. Fix an ordering of the states of $B$ witnessing bandwidth $h$, and let the initial state be the first state in the ordering. We can write the states in this circular order as
\(
    q_0,q_1,\ldots,q_{n-1}.
\) 
Thus every defined transition moves forward, modulo $n$, by a distance between $1$ and $h$. The contradiction will be obtained by identifying a cut in this ordering that is crossed by $k$ distinct $b$-transitions, one from each of the $k$ relevant $b$-paths, even though bandwidth $h=k-1$ allows at most $h$ distinct $b$-transitions to cross any fixed cut.

For $i\in[k]$ and $0\leq t\leq k^2$, let
\[
    s_{i,t}=\delta(q_0,a^ib^t).
\]
These states are all defined, since $a^ib^{k^2}a^i\in L_{f,k}$. We first observe that the states $s_{i,t}$ are pairwise distinct. Indeed, suppose that
\(
    s_{i,t}=s_{j,u}
\) 
for some $i,j\in[k]$ and $0\leq t,u\leq k^2$. Since
\[
    \delta(s_{i,t},b^{k^2-t}a^i)\in F,
\]
we also have
\[
    \delta(s_{j,u},b^{k^2-t}a^i)\in F.
\]
Hence
\[
    a^j b^{u+k^2-t} a^i\in L(B)=L_{f,k}.
\]
By the definition of $L_{f,k}$, this is possible only if $i=j$ and $u=t$. Thus all states $s_{i,t}$ are pairwise distinct. In particular,
\[
    n\geq k(k^2+1).
\]

Consider now the states $s_{i,0}=\delta(q_0,a^i)$, for $i\in[k]$. Since each transition has circular distance at most $h$, the state $s_{i,0}$ is reached from $q_0$ by moving forward by total distance at most $ih$. Moreover,
\[
    ih\leq kh=k(k-1)<k(k^2+1)\leq n,
\]
so no wrap-around (going from \(q_j\) to \(q_j'\), for some $j'\leq j$) can occur while reading any of the prefixes $a^i$. Therefore all states $s_{i,0}$ lie in the interval
\[
    P=\{q_0,q_1,\ldots,q_{kh}\}
\]
of the circular ordering.
Now let $C$ be the (graph) cut immediately after $q_{kh}$. We claim that, for every $i\in[k]$, the path
\[
    s_{i,0},s_{i,1},\ldots,s_{i,k^2}
\]
labelled by $b^{k^2}$ contains a $b$-transition crossing the cut $C$. Indeed, if this path did not cross $C$, then, starting in $P$, it would remain entirely inside $P$. But the path contains $k^2+1$ pairwise distinct states, whereas
\[
    |P|=kh+1=k(k-1)+1=k^2-k+1<k^2+1.
\]
This is impossible, so each of the $k$ paths labelled by $b^{k^2}$ contains a first $b$-transition crossing the cut $C$.

These $k$ crossing transitions are distinct, because all states $s_{i,t}$ are pairwise distinct. On the other hand, in a circular bandwidth-$h$ ordering, at most $h$ distinct $b$-transitions can cross a fixed cut. Indeed, the source state of any transition crossing $C$ must lie among the $h$ states immediately preceding the cut, and by determinism each such state has at most one outgoing $b$-transition. Therefore at most $h=k-1$ distinct $b$-transitions can cross $C$.
However, we have just obtained $k$ distinct $b$-transitions crossing the same cut above, contradicting $h=k-1$. Thus no $(k-1)$-BW-DFA accepts $L_{f,k}$. Therefore
\[
    L_{f,k}\in \mathcal L_{f,k}\setminus \mathcal L_{f,k-1},
\]
and the inclusion is proper.

\end{proof}

In contrast to the results obtained for NFAs, we see that no constant $k$ suffices to obtain all finite languages with $k$-BW-DFAs. Related to the practical motivation for this model, we hence know that we also have to increase the maximal achievable size of a removed factor by co-transcriptional splicing to increase even the number of finite languages that we can express using that operation, if we want to encode deterministic finite automata. For that reason, in the following sections, we now investigate specific questions regarding the required bandwidth to encode given finite languages using $k$-BW-DFAs.

    \section{Testing bandwidth~1 for finite languages}

From potential practical applications' point of view, one would like to minimize the distances between consecutive states in the linear encodings of automata similar to the circular DNA one mentioned in the introduction. The minimal non-trivial bound on those distances is $1$, therefore, deciding whether a language can be encoded with bandwidth $1$ is potentially of high practical interest. Bandwidth $1$ forces the automaton to use a unique single following state to which all outgoing transitions of a state must lead.  In the deterministic setting, the highly restrictive behavior of $1$-BW-DFA leads to the following result.

\begin{theorem}\label{bandwidth1DFA}
Let $M=(Q,\Sigma,\delta,q_0,F)$ be a DFA accepting a finite language.
One can decide in $O(|Q|^3|\Sigma|)$ time whether $L(M)$ is accepted by a
bandwidth $1$ partial DFA.
\end{theorem}

\begin{proof}
First we test whether $L(M)$ is finite. This is done by standard algorithms, deleting states
that are not both reachable and co-reachable, and
checking whether the remaining directed graph has a cycle. If such a cycle
exists, then $L(M)$ is infinite, so we can answer NO. Otherwise the useful
part is acyclic, so from now on we will assume that \(M\) has only useful states and is acyclic.

In such a bandwidth $1$ DFA, from each state $q_i$, every defined transition goes to the unique next state $q_{i+1}$. Since the automaton is acyclic, no transition can wrap from the last state of the circular ordering back to the first, otherwise it would create a directed cycle reachable from \(q_0\) and co-reachable to a final state. Thus, such an automaton is described by
letter sets
\[
A_1,A_2,\ldots,A_N\subseteq\Sigma
\]
and a set $I\subseteq\{0,\ldots,N\}$ of final state positions. It accepts the language
\[
\bigcup_{i\in I} A_1A_2\cdots A_i.
\]

Now let
\[
N=\max\{|w| \mid w\in L(M)\}.
\]
Since the DFA is acyclic, $N\le |Q|-1$, and $N$ can be computed by standard dynamic programming algorithms for
longest path in \(O(|Q|+|\delta|)\) time. Next, for $t=0,\ldots,N$, we define the forward layers
\[
\mathrm{Fwd}[t]=\{\delta(q_0,x) \mid x\in\Sigma^t\}.
\]
and the backward layers 
\[
\mathrm{Back}[t]=\{q\in Q \mid \exists y\in\Sigma^t,\ \delta(q,y)\in F\}.
\]
These sets are computed recursively:
\[
\mathrm{Fwd}[0]=\{q_0\},
\qquad
\mathrm{Fwd}[t+1]=\delta(\mathrm{Fwd}[t],\Sigma),
\]
and
\[
\mathrm{Back}[0]=F,
\qquad
\mathrm{Back}[t+1]=\{q\in Q \mid \exists a\in\Sigma,\ \delta(q,a)\in
\mathrm{Back}[t]\}.
\]
At each recursion level $t$, we can add at most $\min \{|Q|, m|\Sigma|\}$ states to the forward/backward layer, where $m$ is the number of states in the previous layer, so one layer can be computed in $O(|Q||\Sigma|)$ time and there are $O(N)$ layers.
For each $i\in[N]$, define
\[
A_i=\{a\in\Sigma \mid \exists q\in \mathrm{Fwd}[i-1]\text{ such that }
\delta(q,a)\in \mathrm{Back}[N-i]\}.
\]
$A_i$ is exactly the set of letters that occur in position $i$ among
the accepted words of maximum length $N$. For each $i$, we can compute $A_i$ in time $O(|Q||\Sigma|)$. Altogether we get all $A_i$ in time $O(|Q|^2|\Sigma|)$.

Let
\[
I=\{i\in\{0,\ldots,N\} \mid \mathrm{Fwd}[i]\cap F\neq\emptyset\}.
\]

Now construct a partial DFA $B$ with states
\[
0,1,\ldots,N.
\]
The initial state is $0$, and the final states are the positions in $I$.
For $i=1,\ldots,N$, and $a\in A_i$, we set 
$
\delta_{B}(i-1,a)=i
$.
Then
\[
L(B)=\bigcup_{i\in I}A_1A_2\cdots A_i.
\]

We claim that $L(M)$ has a bandwidth $1$ partial DFA
representation if and only if
$
L(M)=L(B).
$

Indeed, if such a representation exists, then every accepted word of maximum length $N$ traverses all $N$ transitions. Hence the $i$-th transition must be carry each letter occurring in position $i$ among the maximum-length words. Therefore the label of the transition is exactly $A_i$. Also, a state $i$ must be final when $L(M)$ contains a word of length $i$, that is, when $i\in I$. Hence, the representation is exactly $B$.

Conversely, if $L(M)=L(B)$, then $B$ itself is
a bandwidth $1$ partial DFA accepting $L(M)$.

Finally, DFA equivalence is decidable in $O(|\Sigma|N^2)$ by standard algorithms. 
\end{proof}

\section{Bounds on the DFA bandwidth of finite languages}
\label{sec:bounds}

In this section we establish some nontrivial upper and lower bounds on the DFA bandwidth of finite languages, with a worst-case multiplicative gap of \(\bigo(D)\) between the bounds, where $D$ is the length of the longest word in the language. The bounds we present can be efficiently computed from the minimal partial DFA accepting the languages. 

The languages \(\emptyset\) and \(\{\varepsilon\}\) have trivially DFA bandwidth $1$. From here on, throughout this section, we assume that the finite languages \(L\) considered contain at least one nonempty word.

\subsection{DFA bandwidth upper bound for finite languages}

Let \(L\subseteq \Sigma^*\) be finite, and let
\(
M=(Q,\Sigma,\delta,q_0,F)
\)
be the (unique) minimal partial DFA recognizing \(L\). As we work with partial DFAs and finite languages, the minimal automata under consideration are acyclic (no sink state needed). This implies that the minimal partial DFA must always have at least one state with no outgoing transitions. All such states must be final, otherwise they could be removed, given that no final state would be reachable from them, contradicting the minimality of \(M\). If there were more than one final state with no outgoing transitions, they could be merged without affecting the languages. Therefore, the minimal partial DFA has at most one state with no outgoing transitions, which is a final state that we call \(f\). As the minimal partial DFA do not have useless states, we can also conclude that \(f\) is reachable from all the other states.

\begin{definition}
For a finite language \(L\), let \(\operatorname{bw}_{\mathrm{DFA}}(L)\)
denote the minimum bandwidth among all partial DFAs recognizing \(L\).
\end{definition}

\begin{theorem}[Position-unfolding upper bound]\label{thm:posupperbound}
Let \(L\subseteq\Sigma^*\) be finite, let \(M\) be the minimal partial DFA
recognizing \(L\), and let \(D=\max\{|w| \mid w\in L\}\). For \(0\le i \le D\), let
\[
H_i=\{\delta(q_0,w) \mid w\in \Sigma^i\text{ and } \delta(q_0,w) \text{ is defined}\},
\qquad h_i=|H_i|,
\]
and define
\[
U(L)=\max_{1\le i\le D}(h_{i-1}+h_i-1).
\]
Then
\[
\operatorname{bw}_{\mathrm{DFA}}(L)\le U(L).
%
\]
\end{theorem}


\begin{proof}
We construct a position-unfolding of the minimal DFA.

For every \(i\ge 0\), let 
\[
V_i=
\left\{
(q,i) \mid
q\in H_i
\right\}.
\]
Thus \(V_i\) consists of the states of the minimal DFA that occur after reading
exactly \(i\) symbols of some accepted word.

Define a partial DFA \(M'\) as follows. Its state set is
\(
V=\bigcup_{i\ge 0} V_i.
\)
The initial state is \((q_0,0)\). Its accepting states are
\[
G=
\left\{
(q,i)\in V_i \mid
q\in F
\right\}.
\]
For a state \((p,i-1)\in V_{i-1}\) and a letter \(a\in\Sigma\), define
\[
\delta'((p,i-1),a)=(q,i)
\]
if and only if in $M$ we have
\(
\delta(p,a)=q.
\)

Because \(M\) is deterministic, \(M'\) is also deterministic. Moreover, \(M'\) accepts the same language as \(M\): 
every accepting computation of \(M\) on a word \(w\in L\)
corresponds to a unique path in \(M'\), and conversely every accepting path in \(M'\) projects to an accepting path in \(M\) at the same sequence of positions.

The automaton \(M'\) is layered in the sense that its state set is the union of 
\(
V_0,V_1,V_2,\ldots,
\)
and every transition goes from \(V_{i-1}\) to \(V_i\) for some \(i\ge 1\).

Now order the states layer by layer:
\[
V_0,\ V_1,\ V_2,\ \ldots
\]
with an arbitrary order inside each layer. Since all transitions go from one layer to the next, this is a topological ordering of the states. Consider a transition from
\(V_{i-1}\) to \(V_i\). In the worst case, its source state is the first state of
\(V_{i-1}\), and its target is the last state of \(V_i\), giving the largest distance in the topological ordering between any pair of states from $V_{i-1}$ and $V_i$, respectively. This distance is at
most
\[
|V_{i-1}|+|V_i|-1.
\]

Since \(|V_i|=h_i\), this layer-by-layer ordering has bandwidth at most
\[
\max_{1\le i\le D}(h_{i-1}+h_i-1)=U(L).
\]
\end{proof}

In what follows let $U(L)$ denote the position-unfolding upper bound on the DFA bandwidth of a finite language $L$, from Theorem~\ref{thm:posupperbound}.

\subsection{DFA bandwidth lower bound for finite languages}

Let \(L\subseteq\Sigma^*\) be finite, and let
\[
M=(Q,\Sigma,\delta,q_0,F)
\]
be the minimal partial DFA recognizing \(L\) (no sink state). 
%
%
Every state \(p\in Q\) of the minimal DFA represents a \emph{residual} language
\[
R_M(p)=\{z\in\Sigma^* \mid \delta(p,z)\in F\}.
\]
If \(p\) is reached from \(q_0\) by a word \(u\), then 
$R_M(p)=u^{-1}L,$
where \(u^{-1}L=\{v \mid uv\in L\}\) is the left-quotient of \(L\) by \(u\). Note that for two words \(u,v\), if \(\delta(q_0,u)=\delta(q_0,v)=p\), then \(u^{-1}L=v^{-1}L=R_M(p)\).

In the minimal partial DFA for \(L\), for \(p\in Q\) and \(\ell\ge 0\), define
\[
B_\ell(p)
=
\left\{
r\in Q \mid
\exists x\in\Sigma^*,\ |x|\le \ell,\ \delta(p,x)=r
\right\},
\]
in words, \(B_\ell(p)\) denotes the ``ball'' of radius \(\ell\) around \(p\).

\begin{theorem}[Residual-ball lower bound]\label{lem:residual-ball}
For every finite language \(L\),
\[
\operatorname{bw}_{\mathrm{DFA}}(L)
\ge
\max_{p\in Q,\ \ell\ge 1}
\left\lceil
\frac{|B_\ell(p)|-1}{\ell}
\right\rceil.
\]
\end{theorem}

\begin{proof}
Let \(A\) be a partial DFA recognizing \(L\), and suppose \(A\) has
bandwidth at most \(k\). To easily distinguish the components of the two DFAs \(M\) and \(A\), we will subscript them with the automaton's name, e.g., \(Q_M\) and \(Q_A\) denote the state set of \(M\) and \(A\), respectively. For a given state \(p_M\) of the minimal DFA accepting \(L\), choose a state \(p_A\) of \(A\) whose residual is the same as \(p_M\)'s. This is always possible, since for any word \(u\) for which \(\delta_M(q_0,u)=p_M\), the state \(p_A=\delta_A(q_0,u)\) has the same residual as \(p_M\), that is, \(R_M(p_M)=R_A(p_A)\).

Every state reachable from \(p_A\) by a word of length at most \(\ell\) must
appear within circular distance at most \(\ell k\) from \(p_A\) in the bandwidth ordering. Let \(n=|Q_M|\). Since \(M\) is minimal, every DFA recognizing \(L\), including \(A\), has at least \(n\) states. If \(\ell k<n\), then \(\ell k<|Q_A|\), so all these states lie in an arc of size at most \(\ell k+1\). If \(\ell k\geq n\), then the inequality \(|B_{\ell}(p_M)|-1\leq \ell k\) is trivial.
On the other hand, the distinct residuals of states in \(B_\ell(p_M)\) require distinct
states of \(A\). Hence
\[
|B_\ell(p_M)|-1\le \ell k.
\]
Therefore
\[
k\ge
\left\lceil
\frac{|B_\ell(p_M)|-1}{\ell}
\right\rceil.
\]
Taking the maximum over \(p_M\) and \(\ell\) gives the result.
\end{proof}

We denote the lower bound in Theorem~\ref{lem:residual-ball} by
\[
\mathrm{LB}_{\mathrm{ball}}(L)
=
\max_{p\in Q,\ \ell\ge 1}
\left\lceil
\frac{|B_\ell(p)|-1}{\ell}
\right\rceil .
\]

\subsection{Gap between bounds} 
Recall from Theorem~\ref{thm:posupperbound}, the notation, for \(i\ge 0\), 
\[
H_i=\{q\in Q \mid \text{ some word of }L\text{ reaches }q
\text{ after exactly }i\text{ letters}\},
\qquad h_i=|H_i|, 
\]
and with \(D=\max\{|w| \mid w\in L\}\), the position-unfolding upper bound
\[
U(L)=\max_{1\le i\le D}(h_{i-1}+h_i-1).
\]

\begin{theorem}
    Let \(L\) be a finite language. Then,
    \[
    \mathrm{LB}_{\mathrm{ball}}(L) 
    \le \operatorname{bw}_{\mathrm{DFA}}(L)
    \le U(L)
    \le 1+(2D-1)\mathrm{LB}_{\mathrm{ball}}(L).
    \]
\end{theorem}
\begin{proof}
    By Theorem~\ref{lem:residual-ball} and Theorem~\ref{thm:posupperbound},
\[
\mathrm{LB}_{\mathrm{ball}}(L)
\le
\operatorname{bw}_{\mathrm{DFA}}(L)
\le
U(L).
\]
We now compare these two quantities. Since every state occurring after exactly
\(i\) letters is reachable from \(q_0\) by a word of length at most \(i\), we have
\[
H_i\subseteq B_i(q_0).
\]
Hence, by the definition of \(\mathrm{LB}_{\mathrm{ball}}(L)\),
\[
h_i\le |B_i(q_0)|
\le
1+i\,\mathrm{LB}_{\mathrm{ball}}(L).
\]
Therefore, for every \(1\le i\le D\),
\[
h_{i-1}+h_i-1
\le
\bigl(1+(i-1)\mathrm{LB}_{\mathrm{ball}}(L)\bigr)
+
\bigl(1+i\mathrm{LB}_{\mathrm{ball}}(L)\bigr)
-1.
\]
Thus
\[
h_{i-1}+h_i-1
\le
1+(2i-1)\mathrm{LB}_{\mathrm{ball}}(L)
\le
1+(2D-1)\mathrm{LB}_{\mathrm{ball}}(L).
\]
Taking the maximum over \(i\) gives
\[
U(L)
\le
1+(2D-1)\mathrm{LB}_{\mathrm{ball}}(L).
\]
\end{proof}

Consequently, the residual-ball lower bound and the position-unfolding upper bound give an \(2D\)-approximation of the DFA bandwidth of finite languages. Moreover, both bounds are polynomial-time computable from the minimal partial DFA \(M\). The sets \(H_i\) required for the upper bound can be computed by dynamic programming following their definition, and then \(U(L)\) is obtained by a maximum search. For the lower bound, it is easy to see that it is enough to consider balls of radius at most \(|Q|-1\), since the shortest path between any two states cannot be longer than that. Then, one can compute \(B_\ell(p)\) for all \(p\in Q\) and \(0\le \ell<|Q|-1\) by dynamic programming.

\begin{example}[Matching bounds]\label{ex:1}
Fix integer \(r\ge 1\)
and let
\[
\Sigma_r=\{a_1,\ldots,a_r,b_1,\ldots,b_r\}.
\]
Define
\[
L_{r}=\{a_jb_j \mid 1\le j\le r\}.
\]

The minimal partial DFA accepting \(L_{r}\) consists of \(r\) paths
\[
q_0
\xrightarrow{a_j}
p_{j}
\xrightarrow{b_j}
f,
\]
one for each \(j=1,\ldots,r\), with the final state \(f\) shared by all paths.

The states \(p_{j}\) have pairwise distinct residuals, since the unique
continuation from \(p_{j}\) is \(b_j\). Thus, the initial
state has \(r\) distinct successor residuals. Hence every partial DFA
recognizing \(L_{r}\) has bandwidth at least \(r\):
\[
\operatorname{bw}_{\mathrm{DFA}}(L_{r})\ge r.
\]

Conversely, order the states as
\[
q_0,
p_{1},\ldots,p_{r},
f.
\]
Every transition is between states at a distance at most \(r\) in the linear ordering above. Thus
\[
\operatorname{bw}_{\mathrm{DFA}}(L_{r})\le r.
\]
Therefore
\[
\operatorname{bw}_{\mathrm{DFA}}(L_{r})=r.
\]

For the position-unfolding upper bound we have 
\(h_0=1\), \(h_1=r\) and \(h_2=1\),
and we get 
\[
U(L_{r})=
\max\{h_0+h_1-1,h_1+h_2-1\}
=r.
\]
For the lower bound, we get
\[
B_1(q_0)=\{q_0,p_{1,1},\ldots,p_{r,1}\},
\]
therefore 
\[
\mathrm{LB}_{\mathrm{ball}}(L_{r})\ge |B_1(q_0)|-1=r.
\]
This, together with 
\(
\mathrm{LB}_{\mathrm{ball}}(L_{r})\le \operatorname{bw}_{\mathrm{DFA}}(L_{r})=r,
\)
implies 
\(
\mathrm{LB}_{\mathrm{ball}}(L_{r})=r.
\)
Thus the bounds meet exactly, for the infinite family \(\{L_{r}\mid r\ge 1\}\), 
that is, 
\[
U(L_{r})=\mathrm{LB}_{\mathrm{ball}}(L_{r})=r, \quad \forall r\ge 1.
\] 
\end{example}

\section{Final Discussion}

We extended the notion of bandwidth for finite automata to the deterministic setting by applying it to partial DFAs. In that context, we have seen that, even for finite languages, a strict infinite hierarchy of language classes defined by bandwidth can be proven (see Figure~\ref{fig:relationship}). Therefore, in contrast to the nondeterministic setting, there does not exist any constant $k$ such that all finite languages can be accepted by $k$-BW-DFAs.
For bandwidth $k=1$, we obtained polynomial time decidability for the question whether a $1$-BW-DFA $A$ with $L(A) = L_f$ exists for a given finite language $L_f$ represented by some DFA. For arbitrary $k$, and DFAs $A$, whether such an efficient decision procedure exists, remains an open problem that could indeed turn out to be NP-hard. Nonetheless, given an arbitrary DFA $A$, we provided efficiently computable upper and lower bounds for the minimal bandwidth $k$ needed to encode $L(A)$ in a $k$-BW-DFA, having a multiplicative gap bounded by $2D$, where $D$ is the length of the longest word in $L(A)$. An example provided showed that there exist input DFAs that hit both bounds simultaneously, resulting in the ability to determine the necessary bandwidth $k$ efficiently. However, the example construction uses an unbounded alphabet that grows with the bandwidth. It is left as future work to determine whether the bounds meet for an infinite family of finite languages over a fixed alphabet.

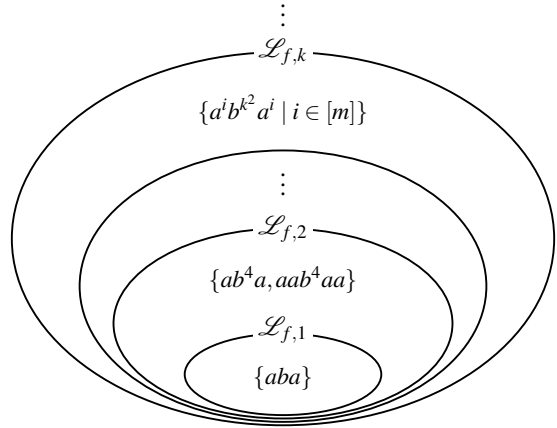
\begin{figure}[h!]
    \centering
    \resizebox{0.45\textwidth}{!}{%
        \begin{tikzpicture}[
            outline/.style={draw, thick},
            levelname/.style={fill=white, inner sep=2pt},
            example/.style={fill=white, inner sep=1.5pt, font=\small}
        ]

        \node[fill=white, inner sep=0pt] at (0, 6.0) {$\vdots$};

        \draw[outline] (0, 2.6) ellipse [x radius=4cm, y radius=2.75cm];
        \node[levelname] at (0, 5.35) {$\mathcal{L}_{f,k}$};
        \node[example] at (0, 4.5)
            {$\{a^i b^{k^2} a^i \mid i\in[m]\}$};

        \draw[outline] (0, 1.9) ellipse [x radius=3cm, y radius=2cm];
        \node[fill=white, inner sep=0pt] at (0, 3.5) {$\vdots$};

        \draw[outline] (0, 1.35) ellipse [x radius=2.5cm, y radius=1.4cm];
        \node[levelname] at (0, 2.75) {$\mathcal{L}_{f,2}$};
        \node[example] at (0, 2.0)
            {$\{ab^4a, aab^4aa\}$};

        \draw[outline] (0, 0.60) ellipse [x radius=1.45cm, y radius=0.60cm];
        \node[levelname] at (0, 1.25) {$\mathcal{L}_{f,1}$};
        \node[example] at (0, 0.55)
            {$\{aba\}$};

        \end{tikzpicture}%
    }
    \caption{A relationship among $\mathcal{L}_{f,1},\mathcal{L}_{f,2},\ldots,\mathcal{L}_{f,k},\ldots$, established in Proposition~\ref{proposition:dfa-bandwidth-finite-hierarchy}.}
    \label{fig:relationship}
\end{figure}

These results extend the solid foundation for relevant questions regarding the optimization of finite automata that are to be encoded on circular DNA for RNA transcription. In particular, they provide a deeper insight if we assume the given automata to be deterministic. In future work, several problems are to be considered. First of all, we would like to know if it can be generally efficiently decided whether a given bandwidth $k$ suffices to encode the language $L(A)$ of a given DFA $A$ in a $k$-BW-DFA. In addition to that, the question of minimization has not been tackled for DFAs, yet. For NFAs and unbounded alphabets, we know from previous work the NP-hardness of the question for a given finite language $L_f$, a bandwidth $k$ and a maximum number of states $n$, whether an $n$-state $k$-BW-NFA $A$ with $L(A)=L_f$ exists or not. Whether such a result also holds for DFAs, remains an open question. Especially for bounded alphabets.

Here we were primarily concerned with finite languages and exact matching, therefore extending this investigation to arbitrary regular languages represented by DFAs or to an approximate setting where some unwanted words with certain properties might still be allowed, could be a worthwhile direction for future research.

Finally let us raise a point about the definition of bandwidth. The well-established notion of bandwidth of a graph is defined slightly differently than how we did, by taking the maximum of the absolute value between the indices of two vertices (states) connected by an edge, over all vertex pairs. We defined the distance modulo the number of states rather than adapting the graph definition, because in the setting where the problem occurred~\cite{ChoFSW25}, the automata were embedded in a circular medium where repeating a transition costs a full circle (hence our definition \(0<(j-i)\bmod n\leq k\)), making it undesirable. However, depending on the one-dimensional medium considered for the encoding, it may make just as much sense to directly adapt the graph definition (\(|j-i|\)), for instance, if one considers a bounded tape with bidirectional head movements for reading the transitions. We think that at least in the case of finite languages (accepted by DFA with acyclic transition graphs) most of our techniques and proof ideas work with the other definition, up to minor changes in the proof details, but leave the treatment of the other naturally defined bandwidth notion as future work, and note that in the case of infinite languages, the differences might be substantial.

\subsection*{Acknowledgements.}
We would like to thank the anonymous referees for their detailed suggestions for improvements, which led, in particular, to a cleaner presentation of the bounds in Section~\ref{sec:bounds} and the correction of Example~\ref{ex:1}.

This work was supported by grants funded by the Korea government (MSIT): the Institute of Information \& Communications Technology Planning \& Evaluation (IITP) under the Artificial Intelligence Convergence Innovation Human Resources Development Program (IITP-2026-RS-2023-00255968) and the National Research Foundation of Korea (NRF) (RS-2025-25436818), awarded to D.-J. C. 
This work was also supported by JSPS KAKENHI Grant Number JP23K10976, awarded to S. Z. F.
In addition, M.~W. gratefully acknowledges Dirk Nowotka of Kiel University for his support and encouragement in participating in this collaboration.

\newpage
\bibliographystyle{eptcs}
\bibliography{bibliography}




\end{document}